\documentclass[letterpaper]{article} % DO NOT CHANGE THIS
\usepackage{aaai25}  % DO NOT CHANGE THIS
\usepackage{times}  % DO NOT CHANGE THIS
\usepackage{helvet}  % DO NOT CHANGE THIS
\usepackage{courier}  % DO NOT CHANGE THIS
\usepackage[hyphens]{url}  % DO NOT CHANGE THIS
\usepackage{graphicx} % DO NOT CHANGE THIS
\usepackage{natbib}  % DO NOT CHANGE THIS AND DO NOT ADD ANY OPTIONS TO IT
\usepackage{caption} % DO NOT CHANGE THIS AND DO NOT ADD ANY OPTIONS TO IT
\usepackage{xcolor}
\usepackage{xspace}
\usepackage{algorithm}
\usepackage{amsmath}
\usepackage{amssymb}
\usepackage{natbib}
\usepackage{mathrsfs} 
\usepackage[noend]{algpseudocode} % algorithmic, \State{}, etc

\usepackage{comment}

\newtheorem{theorem}{Theorem}
\newtheorem{proof}{Sketch Proof}

\newtheorem{remark}{Remark}

\usepackage{newfloat}
\usepackage{listings}
\DeclareCaptionStyle{ruled}{labelfont=normalfont,labelsep=colon,strut=off} % DO NOT CHANGE THIS
\floatstyle{ruled}
\newfloat{listing}{tb}{lst}{}
\floatname{listing}{Listing}
\title{Loosely Synchronized Rule-Based Planning for \\Multi-Agent Path Finding with Asynchronous Actions}
\author{
    Shuai Zhou\textsuperscript{\rm 2}\footnote{Shuai Zhou conducted this research during his internship at UM-SJTU Joint Institute at Shanghai Jiao Tong University.},
    Shizhe Zhao\textsuperscript{\rm 1},
    Zhongqiang Ren\textsuperscript{\rm 1,3}\footnote{Corresponding Author.}
}
\affiliations{
    \textsuperscript{\rm 1}UM-SJTU Joint Institute, Shanghai Jiao Tong University, China\\
    \textsuperscript{\rm 2}SHIEN-MING WU School of Intelligent Engineering, South China University of Technology, China\\
    \textsuperscript{\rm 3}Department of Automation, Shanghai Jiao Tong University, China
    davidzhou718@gmail.com, \{shizhe.zhao,zhongqiang.ren\}@sjtu.edu.cn
}

\usepackage{bibentry}
\newcommand{\hidecomments}[1]{}

\newcommand{\algorithmicpreface}{\textbf{Notation:}}
\newcommand{\Notation}{\item[\algorithmicpreface]}
\newcommand{\asypibt}{\ensuremath{\rm ASY\text{-}PUSH}\xspace}
\newcommand{\asypibtswap}{\ensuremath{\rm ASY\text{-}PUSH\text{-}SWAP}\xspace}
\newcommand{\masipp}{SIPP\xspace}
\newcommand{\pullop}{\rm PULL\xspace}
\newcommand{\pushop}{\rm PUSH\xspace}

\newcommand{\swapop}{\rm SWAP\xspace}
\begin{document}

 		\thispagestyle{plain}
		  \pagestyle{plain}
            \pagenumbering{arabic}

\maketitle

\begin{abstract}
Multi-Agent Path Finding (MAPF) seeks collision-free paths for multiple agents from their respective starting locations to their respective goal locations while minimizing path costs. Although many MAPF algorithms were developed and can handle up to thousands of agents, they usually rely on the assumption that each action of the agent takes a time unit, and the actions of all agents are synchronized in a sense that the actions of agents start at the same discrete time step, which may limit their use in practice. Only a few algorithms were developed to address asynchronous actions, and they all lie on one end of the spectrum, focusing on finding optimal solutions with limited scalability. This paper develops new planners that lie on the other end of the spectrum, trading off solution quality for scalability, by finding an unbounded sub-optimal solution for many agents.
Our method leverages both search methods (LSS) in handling asynchronous actions and rule-based planning methods (PIBT) for MAPF. We analyze the properties of our method and test it against several baselines with up to 1000 agents in various maps.
Given a runtime limit, our method can handle an order of magnitude more agents than the baselines with about 25\% longer makespan.

\end{abstract}


\begin{links}
    \link{Code}{https://github.com/rap-lab-org/public_LSRP}
    % \link{Datasets}{https://aaai.org/example/datasets}
    \link{Extended version}{https://arxiv.org/abs/2412.11678}
\end{links}

\section{Introduction}

Multi-Agent Path Finding (MAPF) computes collision-free paths for multiple agents from their starting locations to destinations within a shared environment, while minimizing the path costs, which arises in applications such as warehouse logistics.
Usually the environment is represented by a graph, where vertices represent the location that the agent can reach, and edges represent the transition between two locations.
MAPF is NP-hard to solve to optimality~\cite{yu2013structure}, and a variety of MAPF planners were developed, ranging from optimal planners~\cite{sharon2015conflict,wagner2015subdimensional}, bounded sub-optimal planners~\cite{barer2014suboptimal,li2021eecbs} to unbounded sub-optimal planners~\cite{okumura2022priority,de2013push}.
These planners often rely on the assumption that each action of any agent takes the same duration, i.e., a time unit, and the actions of all agents are synchronized, in a sense that, the action of each agent starts at the same discrete time step.
This assumption limits the application of MAPF planners, especially when the agent speeds are different or an agent has to vary its speed when going through different edges.

To get rid of this assumption on synchronous actions, MAPF variants such as Continuous-Time MAPF~\cite{ANDREYCHUK2022103662}, MAPF with Asynchronous Actions~\cite{ren2021loosely}, MAPF$_R$~\cite{walker2018extended} were proposed, and only a few algorithms were developed to solve these problems.
% All 
Most of these algorithms lie on one end of the spectrum, finding optimal or bounded sub-optimal solutions at the cost of limited scalability as the number of agents grows.
To name a few, Continuous-Time Conflict-Based Search (CCBS)~\cite{ANDREYCHUK2022103662} extends the well-known Conflict-Based Search (CBS) to handle various action times and is able to find an optimal solution to the problem.
Loosely Synchronized Search (LSS)~\cite{ren2021loosely} extends A* and M*~\cite{wagner2015subdimensional} to handle the problem and can be combined with heuristic inflation to obtain bounded sub-optimal solutions.
Although these algorithms can provide solution quality guarantees, they can handle only a small amount of agents ($\leq 100$) within a runtime limit of few minutes.
Currently, we are not aware of any algorithm that can scale up to hundreds of agents with asynchronous actions.
This paper seeks to develop new algorithms that lie on the other end of the spectrum, trading off completeness and solution optimality for scalability.

When all agents' actions are synchronous, the existing MAPF algorithms, such as rule-based planning~\cite{erdmann1987multiple,luna2011push,de2013push,okumura2022priority}, can readily scale up to thousands of agents by finding an unbounded sub-optimal solution.
However, extending them to handle asynchronous actions introduce additional challenges:
Rule-based planning usually relies on the notion of time step where all agents take actions and plans forward in a step-by-step fashion.
When the actions of agents are of various duration, there is no notion of planning steps, and one may have to plan multiple actions for a fast moving agent and only one action for a slow agent.
In other words, an action with long duration of an agent may affect multiple subsequent actions of another agent, which thus complicates the interaction among the agents.
To handle these challenges, we leverage the state space proposed in~\cite{ren2021loosely} where times are incorporated, and leverage Priority Inheritance with Backtracking (PIBT)~\cite{okumura2022priority}, a recent and fast rule-based planner, to this new state space.
We introduce a cache mechanism to loosely synchronize the actions of agents during the search, when these actions have different, yet close, starting times.
We therefore name our approach Loosely Synchronized Rule-based Planning (LSRP).

We analyze the theoretic properties of LSRP and show that LSRP guarantees reachability in graphs when the graph satisfies certain conditions.
For the experiments, we compare LSRP against several baselines including CCBS~\cite{ANDREYCHUK2022103662} and prioritized planning
in various maps from a MAPF benchmark \cite{stern2019multi}, and the results show that LSRP can solve up to an order of magnitude more agents than existing methods with low runtime,
despite about 25\% longer makespan.
Additionally, our asynchronous planning method produces better solutions, whose makespan ranges from 55\% to 90\% of those planned by ignoring the asynchronous actions.

\section{Problem Definition}\label{lsrp:probDefine}
Let set $I =\{1,2,\dots,N\}$ denote a set of $N$ agents.
All agents move in a workspace represented as a finite graph $G =(V,E)$, where the vertex set $V$ represents all possible locations of agents and the edge set $E \subseteq V \times V$ denotes the set of all the possible actions that can move an agent between a pair of vertices in $V$.
An edge between $u, v \in V$ is denoted as $(u,v)\in E$ and the cost of $e \in E$ is a finite positive real number $cost(e) \in \mathbb{R}^{+}$.
Let $v_s^i, v_g^i \in V$ respectively denote the start and goal location of agent $i$.

All agents share a global clock and start moving from $v_s^i$ at $t$ = 0. 
Let $D(i,v_1,v_2) \in \mathbb{R}^+$ denote amount of time (i.e., duration) for agent $i$ to go through edge $(v_1,v_2)$.
Note that different agents may have different duration when traversing the same edge, and the same agent may have different duration when traversing different edges.\footnote{
As a special case, a self-loop $(v,v), v\in V$ indicates a wait-in-place action of an agent and its duration $D(i,v,v)$ is the amount of waiting time at vertex $v$, which can be any non-negative real number and is to be determined by the planner.}

When agent $i$ goes through an edge $(v_1,v_2) \in E$ between times $(t_1, t_1 + D(i,v_1,v_2))$, agent $i$ is defined to occupy: (1) $v_1$ at $t = t_1$, (2) $v_2$ at $t = t_2$ and (3) both $v_1$,$v_2$ within the open interval  $(t_1, t_1 + D(i,v_1,v_2))$.
Two agents are in conflict if they occupy the same vertex at the same time.
We refer to this definition of conflict as the \emph{duration conflict} hereafter.
Fig.~\ref{fig:collision_model} provides an illustration.

\begin{figure}[tb]
    \centering
    \includegraphics[width=0.5\linewidth]{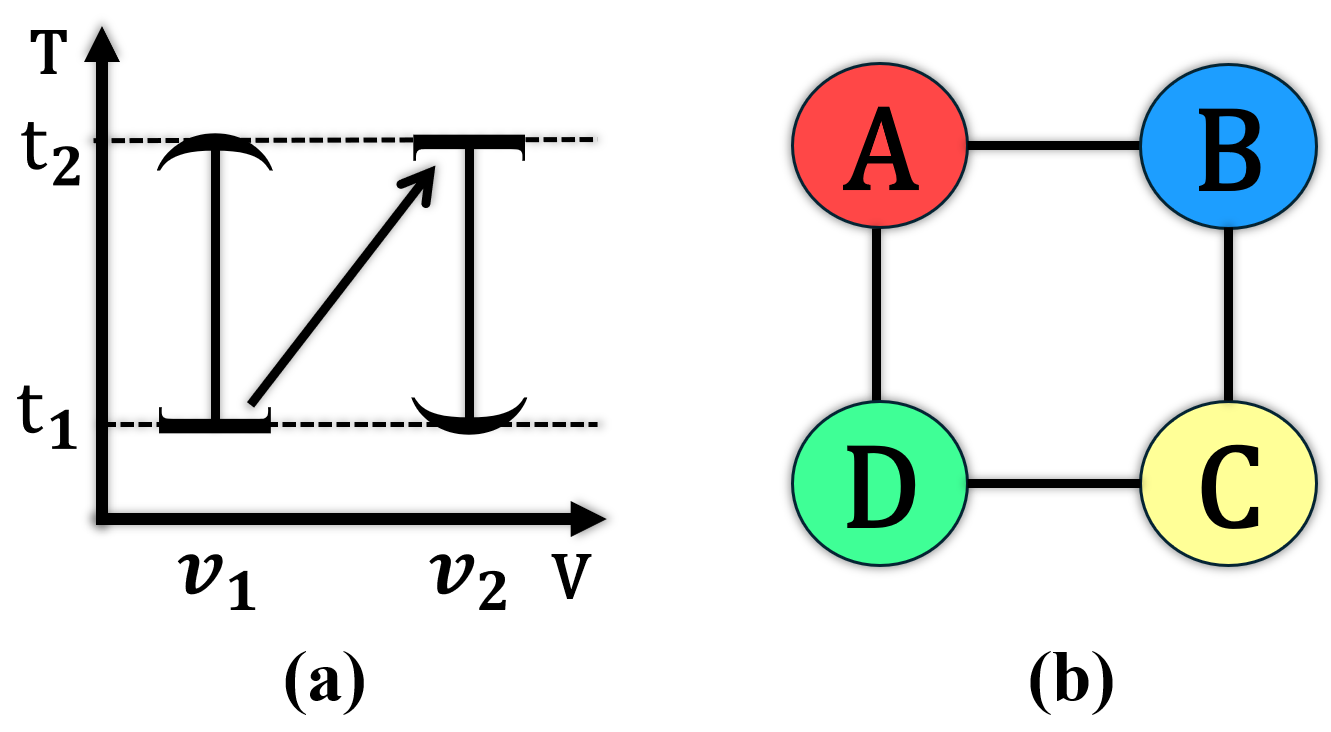}
    \caption{(a) shows the occupation of vertices when an agent traverses edge $(v_1,v_2)$ between times $(t_1,t_2)$ as shown by the black arrow. Round brackets represent open intervals, while square brackets represent closed intervals. The black vertical lines mean the vertices are occupied during the time range. (b) By duration conflict, the agents cannot move, while in conventional MAPF~\cite{stern2019multi}, the agents can move together clockwise or counter-clockwise.}
    \label{fig:collision_model}
\end{figure}

Let $\pi^{i}$ denote a path from $v_s^i$ to $v_g^i$ via a sequence of vertices $v \in G$. Any two vertices $v_{k}^i$ and $v_{k+1}^i$ in $\pi^{i}$ are either connected by edge $(v_{}^i,v_{+1}^i) \in E$ or is a self-loop.
Let $g(\pi^{i}(v_s^i,v_g^i))$ denote the cost of the path, which is defined as the sum of duration of edges along the path.
Let $\pi = (\pi^1,\pi^2,\dots,\pi^n)$ represent a joint path of all agents, and its cost is the sum of individual path costs of all the agents, i.e., $g(\pi) = \sum_{i}g(\pi^{i})$.

The goal of the Multi-Agent Path Finding with Asynchronous Actions (MAPF-AA) is to find a conflict-free joint path $\pi$ connecting $v_s^i,v_g^i$ for all agents $i \in I$, such that $g(\pi)$ reaches the minimum. This work seeks to develop algorithms that can quickly solve MAPF-AA instances with many agents by finding a (unbounded) sub-optimal solution.

\begin{remark}
    The same definition of occupation and conflict was introduced in \cite{ren2021loosely}, which generalizes the notion of following conflict and cycle conflict in \cite{stern2019multi}, and is similar to the ``mode'' in \cite{okumura2021time}.
    Conventional MAPF usually considers vertex and edge conflicts~\cite{sharon2015conflict,okumura2023lacam2}, which differ from the duration conflict here.
    Another related problem definition is MAPF with duration conflict (denoted as MAPF-DC), which replaces the vertex and edge conflict in MAPF with duration conflict.
    MAPF-DC is a special case of MAPF-AA where the duration is the same constant number for any agent and any edge.
\end{remark}

\section{Preliminaries}
\subsection{Priority Inheritance with Backtracking} \label{lsrp:pibt} 
Priority Inheritance with Backtracking (PIBT) plans the actions of the agents in a step-by-step manner until all agents reach their goals.
PIBT assigns each agent a changing priority value.
In each step, a planning function is called to plan the next action of the agents based on their current priorities.
This planning function selects actions based on the individual shortest path to the goal of each agent, and actions toward a location closer to the goal are first selected.
When two agents seek to occupy the same position, the higher-priority agent is able to take this location, and pushes the lower-priority agent to another less desired location. This function is applied recursively, where the pushed agent is planned next and inherits the priority of the pushing agent.
When all agents' actions are planned for the current time step, PIBT starts a new iteration to plan the next time step.

PIBT guarantees that the agent with the highest priority eventually reaches its goal, at which it becomes the lowest priority agent. Therefore, each agent becomes the highest priority agent at least once and is able to reach its goal at some time step.
PIBT requires that for each vertex $v\in G$, there is a cycle in $G$ containing $v$, so that PIBT can plan all agents to their goals.
Otherwise, PIBT is incomplete, i.e., PIBT may not be able to find a feasible solution even if the instance is solvable.
PIBT runs fast and can scale to many agents for MAPF.
Our LSRP leverages the idea of PIBT to handle a large number of agents.

\subsection{Loosely Synchronized Search} \label{lsrp:lss}
Loosely Synchronized Search (LSS) extends A* and M*-based approaches to solve MAPF-AA by introducing new search states that include both the locations and the action times (i.e., as timestamps) of the agents.
Similarly to A*, LSS iteratively selects states from an open list, expands the states to generate new states, prunes states that are in-conflict or less promising, and adds remaining states to open for future expansion, until a conflict-free joint path for all agents is found from the starts to the goals.
To expand a state, LSS only considers the agent(s) $i\in I$ with the smallest timestamps and plan its actions, which increases the timestamp of agent $i$.
In a future iteration, other agents $j\neq i$ will be planned if the timestamp of $j$ becomes the smallest.
Planning all agents together may lead to a large branching factor and LSS leverage M* to remedy this issue.
LSS is complete and finds an optimal solution for MAPF-AA but can only handle a relatively small amount of agents.
Our LSRP leverages the state definition and expansion in LSS to handle asynchronous actions.

\subsection{Other Related Approaches}\label{lsrp:other-related}
Safe Interval Path Planning (SIPP)~\cite{phillips2011sipp} is a single-agent graph search algorithm that can find a path from start to goal with the minimum arrival time among dynamic obstacles along known trajectories.
SIPP can be used together with priority-based planning to handle MAPF-AA.
Specifically, each agent is assigned with a unique priority, and the agents are planned from the highest priority to the lowest using SIPP, where the planned agents are treated as dynamic obstacles.
This priority-based method is used as a baseline in our experiments.

Additionally, CCBS~\cite{ANDREYCHUK2022103662} is a two-level search algorithm that can be used to handle MAPF-AA.
CCBS is similar to CBS~\cite{sharon2015conflict} for MAPF.
The high-level search detects conflicts between any pair of agents, and resolves conflicts by generating constraints that forbid an agent from using certain vertices within certain time ranges.
The low-level search uses SIPP to plan a single-agent path subject to the constraints added by the high-level.
CCBS iteratively detects conflicts on the high-level and resolves conflicts using the low-level search until no conflict is detect along the paths.
CCBS is guaranteed to find an optimal solution if the given problem instance is solvable.
In practice, CCBS can handle tens of agents within a few minutes~\cite{ANDREYCHUK2022103662}.
This paper uses CCBS as another baseline in the experiments.

\section{Method}\label{lsrp:lsrp}

\begin{algorithm}[tb]
\caption{LSRP, {\underline{LSRP-SWAP}}}

\small
\begin{algorithmic}[1]
\Require $ graph\; \mathcal{G},starts \; \{{ v_s^1},\ldots,{ v_s^n}\}, { goals}\; \{{ v_g^1},\ldots,{ v_g^n}\}$ 
\Ensure ${ paths}\; \{{ \pi^1}, \ldots,{ \pi^n}\}$ 
% \Notation $\epsilon^i : {\rm Initial\, priority\, of\, agent \,i}$
\State ${ T \gets \{0\}; \; S_{T} \gets \{s_0\}; \; \Phi \gets \; \{ \}}$
% \State $\forall i\in I: S_{T}[0][i] \gets state(s^i, s^i, 0, 0)$ 
\State $\epsilon_0 \gets \textsc{InitPriority}()$ \label{alg:initialpri}
\State $\epsilon \gets \epsilon_0$

\While{$T \neq \emptyset$}
\State ${ s_{prev} \gets S_{T}.back()}$\label{alg:prev}
\If{$\forall i\in I, s_{prev}^i.v=v_g^i$}
\State \Return \textsc{PostProcess}($S_{T}$) \label{alg:postpro}
\EndIf
\For{$i \in I$}
\State\textbf{if}{ $ s_{prev}^i.v = v^i_g$ } \label{alg:resepri} \textbf{then} $\epsilon^i \gets \epsilon^i_0$
\State\textbf{else}  $\epsilon^i \gets \epsilon^i + 1$ \label{alg:resepriend}
\EndFor
\State $ t_{min}\gets T.pop()$
\State $ I_{curr}\gets $ \textsc{ExtractAgents}($I$, $t_{min}$, $s_{prev}$) \label{alg:extract}
\If{$T \neq \emptyset$} \label{alg:tnextbegin}
\State $t_{next} \gets T.top()$ \label{alg:ttop}
\Else
\State $t_{next} \gets t_{min} + \min_{i\in I, e\in E}{D(i,e)}$ \label{alg:tnextend}
\EndIf
\State $s_{next} \; \gets \textsc{Get\_Snext}(\Phi,I_{curr},s_{prev}) $ \label{alg:snext}
\For{$i \in I_{curr} \text{ in descending 
  order of } \epsilon^i$}
 \label{alg:Icurr}
    \If{$s_{next}^i = \emptyset$} \label{alg:nocache}
    \State {${\rm ASY\text{-}PUSH}(i, \{\},t_{min},t_{next},{\rm False})$}
    \State {\underline{(or ${\rm ASY\text{-}PUSH\text{-}SWAP}(i, \{\},t_{min},t_{next},{\rm False})$})}\label{alg:line:asy-push-swap}
    \EndIf
\EndFor
\State $S_{T}.append(s_{next})$
\State ${\rm Add \,} s_{prev}^i.t_v\ {\rm for\,} i \in I {\rm\ to\,}\ T$
\EndWhile
\State \textbf{return} failure
\end{algorithmic}\label{lsrp:alg:lsrp}
\end{algorithm}

\subsection{Notation and State Definition} \label{lsrp:sec:notation}
Let $\mathcal{G}=(\mathcal{V},\mathcal{E}) = G \times G \times \dots \times G$ denote the joint graph, the Cartesian product of $N$ copies of $G$, where $v \in \mathcal{V}$ represents a joint vertex, and $e \in \mathcal{E}$ represents a joint edge that connects a pair of joint vertices.
The joint vertices corresponding to the start and goal vertices of all the agents are $v_{s} = (v_{s}^1,v_{s}^2,\cdots,v_{s}^n)$ and $v_{g} = (v_{g}^1,v_{g}^2,\cdots,v_{g}^n)$ respectively. A joint search state~\cite{ren2021loosely} is $s = (s^1,s^2,\dots,s^n)$, where $s^i$ is the individual state of agent $i$, which consists of four components: (1) $p \in V$, a (parent) vertex in $G$, from which the agent $i$ begin its action; (2) $v \in V$, a vertex in $G$, at which the agent $i$ arrives; (3) $t_{p} $, the timestamp of $p$, representing the departure from $p$; (4) $t_{v} $, the timestamp of $v$, representing the arrival time at $v$.

An individual state $s^i = (p,v,t_{p},t_{v})$ describe the location occupied by agent $i$ within time interval $[t_{p},t_s]$ with a pair of vertices $(p,v)$.
Intuitively, an individual state is also an action of agent $i$, where $i$ moves from vertex $p$ to $v$ between timestamps $t_{p}$ and $t_{v}$.
For the initial state $s_0$, we define $p = v = v_{s}^i$ and $t_{p} = t_{v} = 0, \forall i \in I$.
Let $\epsilon = (\epsilon^1,\epsilon^2,\cdots,\epsilon^N)$ denote the priorities of the agents, which are positive real numbers in $[0,1]$.
In this paper, we use the dot operator ($.$) to retrieve the element inside the variable (e.g. $s^i.t_v$ denote the arrival time $t_v$ of agent $i$ in the individual state $s^i$).

\subsection{Algorithm Overview}\label{lsrp:sec:algo_overview}

LSRP is shown in Alg.~\ref{lsrp:alg:lsrp}.
Let $T$ denote a list of timestamps, and LSRP initializes $T$ as a list containing only $0$, the starting timestamp of all agents.
Let $S_T$ denote a list of joint states, which initially contains only the initial joint state.
Let $\Phi$ denote a dictionary where the keys are the timestamps and the values are joint states.
$\Phi$ is used to cache the planned actions of the agents and will be explained later.
Then, LSRP initializes the priorities of the agents (Line \ref{alg:initialpri}), which can be set in different ways (such as using random numbers).

LSRP plans the actions of the agents iteratively until all agents reach their goals.
LSRP searches in a depth-first fashion.
In each iteration, LSRP retrieves the most recent joint state that was added to $S_T$, and denote it as $s_{prev}$ (Line \ref{alg:prev}).
If all agents have reached their goals in $s_{prev}$, $S_T$ now stores a list of joint states that brings all agents from $v_s$ to $v_g$. LSRP thus builds a joint path out of $S_T$, which is returned as the solution (Line \ref{alg:postpro}), and then terminates.
Otherwise, LSRP resets the priorities of the agents that have reached their goals, and increase the priority by one for agents that have not reached their goals yet (Lines \ref{alg:resepri}-\ref{alg:resepriend}).

Then, LSRP extracts the next planning timestamp $t_{min}$ from $T$, which is the minimum timestamp in $T$, and extracts the subset of agents $I_{curr} \subseteq I$ so that for each agent $i \in I$, its arrival timestamp $t^i_v$ in $s_{prev}$ is equal to $t_{min}$ (Line~\ref{alg:extract}).
Intuitively, the agents in $I_{curr}$ needs to determine their next actions at time $t_{min}$.
Here, $t_{next}$ is assigned to be the next planning timestamp in $T$ if $T$ is not empty.
Otherwise ($T$ is empty), $t_{next}$ is set to be $t_{min}$ plus the a small amount of time (Lines \ref{alg:tnextbegin}-\ref{alg:tnextend}).
To generate the next joint state based on $s_{prev}$, LSRP first checks if the actions of agents in $I_{curr}$ have been planned and cached in $\Phi$.
If so, these cached actions are used and the corresponding individual states are generated and added to $s_{next}$ (Line \ref{alg:snext}).
For agents without cached actions (Line \ref{alg:nocache}), LSRP invokes a procedure \asypibt to plan the next action for that agent.

Finally, the generated next joint state $s_{next}$ is appended to the end of $S_T$, and the arrival timestamp $s^i_{prev}.t_v$ of each agent $i \in I$ is added to $T$ for future planning.

\subsection{Recursive Asynchronous Push}\label{lsrp:sec:async_pibt}

\asypibt takes the following inputs: $i$, the agent to be planned; $ban$, a list of vertices that agent $i$ is banned from moving to; $t$, the current timestamp; $t_{next}$, the next timestamp; and $bp$, a boolean value indicating if agent $i$ is being pushed away by other agents, which stands for ``be pushed''.

\begin{algorithm}[tb]

\small

\caption{$\rm ASY\text{-}PUSH, \underline{ASY\text{-}PUSH\text{-}SWAP}$}
\begin{algorithmic}[1]
\Require $i, ban,t,t_{next},bp$
\Notation $v^i \leftarrow s_{prev}^i.v$
    \State $C \leftarrow \text{Neigh}(v^i) \cup \{v^i\}$ \label{alg2:neigh}
    \State sort $C$ in increasing order of $\text{dist}(u, v_g^i)$ where $u \in C$ \label{alg2:sortneigh}
    \State $\underline{ j \leftarrow \textsc{Swap-Required-possible}(i,C[0])} $ \label{alg2:swap}
    \State \underline{\textbf{if }$ j \neq \emptyset {\rm \textbf{ then }}C.reverse()$} \label{alg2:creverse}
    \If{$\epsilon^i > \epsilon^n (n \neq i) \forall n \in I$} \label{alg2:reachbility}
    \State $C.move(v^i, 1)$ \label{alg2:reachend}\Comment{Move $v^i$ to second}
    \EndIf

    \For{$v \in C$} \label{alg2:main}
        \State \textbf{if} ${  \textsc{Occupied}(v,s_{next},ban,bp)}{\rm \textbf{ then continue}}$ \label{alg2:occupied}
        \State ${k \leftarrow \textsc{Push-Required}()}$ \label{alg2:push_required}
        \If{${k \neq \emptyset}$}
            \State $ban.append(v^i)$
            \State ${  t_{wait}^i\leftarrow\; \asypibt(k,ban,t,t_{next},{\rm True})}$ \label{alg2:asy_pibt}
            \State ${\textbf{if} \;t_{wait}^i = \emptyset\;\textbf{then continue}}$\label{alg2:pushfail}
            \State${ \textsc{WaitAndMove}(i,v,t_{wait}^i,s_{next},\Phi)}$ \label{alg2:phi}
            \State ${t_{move}^i \leftarrow t_{wait}^i + D(i,v^i,v)}$
            \If{$ \underline{ (!bp) \wedge v = C[0] \wedge j \neq \emptyset \wedge s_{next}^j = \emptyset}$} \label{alg2:kexswap}
            \State$\underline{ \textsc{WaitAndMove}(j,v^i,t_{move}^i,s_{next},\Phi)}$ \label{alg2:kexswapend}
            \EndIf
            \State \Return $t_{move}^i$ \label{alg2:kexend}
        \EndIf    
         \If {$v = v^i$} \label{alg2:nokexist}
        \State ${ s_{next}^i \;\leftarrow \;state(v^i,v,t,t_{next})}$
        \State \Return \label{alg2:waitend}
        \EndIf
        \State $t_{move}^i \leftarrow t + D(i,v^i,v)$ 
         \State ${ s_{next}^i \;\leftarrow \;state(v^i,v,t,t_{move}^i)}$ \label{alg2:snext generated}
         \If{$ \underline{ (!bp) \wedge v = C[0] \wedge j \neq \emptyset \wedge s_{next}^j = \emptyset}$} \label{alg2:knexswap}
            \State$\underline{ \textsc{WaitAndMove}(j,v^i,t_{move}^i,s_{next},\Phi)}$ \label{alg2:knexswapend}
        \EndIf
        \State \Return $t_{move}^i$ \label{alg2:async_end}
    \EndFor
    \State \Return $\emptyset$  \label{alg2:unpushable}
\end{algorithmic}
\label{lsrp:alg:asy_pibt}
\end{algorithm} 

\begin{algorithm}[tb]

% \small

\caption{$\textsc{WaitAndMove}$}
    \begin{algorithmic}[1]
    \Require $i,v,t_{wait}^i,s_{prev},s_{next},\Phi$
    \Notation $v^i\leftarrow s_{prev}^i.v$
    \State ${s_{next}^i \;\leftarrow\;state(v^i,v^i,t,t_{wait}^i)}$ \label{alg3:wait}
    \State ${t_{move}^i \leftarrow t_{wait}^i + D(i,v^i,v)}$\label{alg3:caculate}
    \State Add ${ state(v^i,v,t_{wait}^i,t_{move}^i)}$ to $\Phi$ \label{alg3:store}
    \State \Return
    \end{algorithmic} \label{lsrp:alg:waitandmove}
\end{algorithm}
At the beginning, \asypibt identifies all adjacent vertices $C$ that agent $i$ can reach from its current vertex in $G$.
These vertices in $C$ are sorted based on its distance to the agent's goal (Line \ref{alg2:neigh}-\ref{alg2:sortneigh}, Alg.\ref{lsrp:alg:asy_pibt}).
Then, for each of these vertices from the closest to the furthest, \asypibt checks whether the agent can move to that vertex without running into conflicts with other agents, based on the occupancy status of that vertex (Line \ref{alg2:main}-\ref{alg2:unpushable}, Alg.\ref{lsrp:alg:asy_pibt}).
\asypibt stops as soon as a valid vertex (i.e., a vertex that is unoccupied or can be made unoccupied through the push operation) is found.
Specifically, between lines \ref{alg2:main}-\ref{alg2:unpushable} in Alg.\ref{lsrp:alg:asy_pibt}, \asypibt may run into one of the following three cases:

    \noindent\textbf{Case 1} Occupied (Line \ref{alg2:occupied}): The procedure \textsc{Occupied} returns true if either one of the following three conditions hold. (1) Vertex $v$ is inside $ban$, which means agent $i$ cannot be pushed to $v$. (2) $v$ is occupied by another agent $i'$ and $i'$ either has been planned or $i'$ is not in $I_{curr}$. (3) $bp$ is true (which indicates that agent $i$ is to be pushed away) and $v$ is the vertex currently occupied by agent $i$ (i.e., $v = s^i_{prev}.v$).
    When \textsc{Occupied} returns true, agent $i$ cannot move to $v$. Then, \asypibt ends the current iteration of the for-loop, and checks the next vertex in $C$.
    
    % \item 
    \noindent\textbf{Case 2} Pushable (Line \ref{alg2:push_required}-\ref{alg2:kexend}): \asypibt invokes \textsc{Push-Required} to find if vertex $v$ is occupied by another agent $k$ that satisfy the following two conditions: (1) $k \in I_{curr}$ and (2) the action of $k$ has not yet been planned (Line \ref{alg2:push_required}).
    If no such a $k$ is found, \asypibt goes to the ``Unoccupied'' case as explained next.
    If such a $k$ is found, the current vertex $v^i$ occupied by agent $i$ is added to the list $ban$ so that agent $k$ will not try to push agent $i$ in a future recursive call of \asypibt, which can thus prevent cyclic push in the recursive \asypibt calls.
    Then, a recursive call of \asypibt on agent $k$ is invoked and the input argument $bp$ is marked true, indicating that agent $k$ is pushed by some other agent.
    $bp$ is also used in SWAP-related procedures (Line \ref{alg2:swap}-\ref{alg2:creverse},\ref{alg2:kexswap}-\ref{alg2:kexswapend} and \ref{alg2:knexswap}-\ref{alg2:knexswapend}), which will be explained later.
    This recursive call returns the timestamp when agent $k$ finished its next action, and agent $i$ has to wait till this timestamp, which is denoted as $t^i_{wait}$.
    Given $t^i_{wait}$, the procedure \textsc{WaitAndMove} is invoked to add both the wait action and the subsequent move action of agent $i$ into $\Phi$, the dictionary storing all cached actions.
    Then, the timestamp $t^i_{move}$ when agent $i$ reaches $v$ is computed (Line~\ref{alg2:kexend} ) and returned.

    \noindent\textbf{Case 3} Unoccupied (Line \ref{alg2:nokexist}-\ref{alg2:async_end}): Vertex $v$ is valid for agent $i$ to move into and a successor individual state $s_{next}^i$ for agent $i$ is generated (Line~\ref{alg2:snext generated}). Then, \asypibt returns the timestamp $t^i_{move}$ when agent $i$ arrives at $v$ (Line \ref{alg2:async_end}).

\subsubsection{Cache Future Actions}
During the search, $\Phi$ caches the planned actions of the agents, and is updated in Alg.\ref{lsrp:alg:waitandmove}, which takes an agent $i\in I$, a vertex $v\in V$, the timestamp $t^i_{wait}$ that agent $i$ needs to wait before moving as the input.
Alg.\ref{lsrp:alg:waitandmove} first generates the corresponding individual state $s_{next}^i$ where agent $i$ waits in place till $t^i_{wait}$(Line \ref{alg3:wait}), and then calculates time $t^i_{move}$ when agent $i$ reaches $v$ after the wait (Line \ref{alg3:caculate}).
Finally, the future individual state corresponding to the move action of agent $i$ from $v^i$ to $v$ between timestamps $[t^i_{wait}, t^i_{move}]$ is generated and stored in $\Phi$.

\subsubsection{Toy Example}
 Fig.~\ref{fig:Lsrp} shows an example in an undirected graph with three agents $I=\{1,2,3\}$ corresponding to the yellow, blue, and red circles. The duration for the agents to go through any edge is 1, 2, 3 respectively. $S_{0}$ shows the initial individual states.
 The initial priorities of the agents are set to $\{0.99, 0.66, 0.33\}$ respectively.
 At timestamp $t_{min} = 0$, $I_{curr}$ includes all three agents.
 LSRP calls \asypibt based on their priorities from the largest to the smallest.

LSRP calls \asypibt with $i=1$, and agent $1$ attempts to move to vertex D as D is the closest to its goal. $D$ is now occupied by agent $2$, which leads to a recursive call on \asypibt with $i=2$ to check if agent $2$ can be pushed away (Line \ref{alg2:asy_pibt}, Alg.\ref{lsrp:alg:asy_pibt}). In this recursive call on $i=2$, agent $2$ tries to move to vertex B as B is the closest to its goal, which is occupied by agent $3$, and another recursive call on \asypibt with $i=3$ is conducted.
In this recursive call on \asypibt with $i=3$, agent $3$ finds its goal vertex C is unoccupied, and $s_{next}^3 = (B, C, 0, 3)$ is generated, which is shown in Fig.~\ref{fig:Lsrp}(b). Here, the returned timestamp is $t_{move}^3 = 3$, which is the timestamp when agent $3$ arrives vertex C.
Then, the call on \asypibt with $i=2$ gets this returned timestamp and set it as $t^2_{wait} = 3$, the timestamp that agent $2$ needs to wait till before moving to vertex D.
An individual state $s_{next}^2=(D,D,0,3)$ is generated, which is $s^2_{next}$ in Fig.~\ref{fig:Lsrp}(b), and another individual state $(D,B,3,5)$ is created and stored in $\Phi$, which corresponds to the move action of agent $2$ from $D$ to $B$.
Finally, the call on \asypibt with $i=1$ receives $t_{wait}^1 = 5$, which is the timestamp that agent $1$ has to wait till before moving to $D$.
An individual state $s_{next}^1=(E,E,0,5)$ is generated, corresponding to the wait action of agent $1$, which is $s^1_{next}$ in Fig.~\ref{fig:Lsrp}(b).
Then another individual state $(E,D,5,6)$ is also created and stored in $\Phi$ as the future move action of agent $1$.
Finally, all three agents are planned and $s_{next}$ is added to  $S_T$.
The arrival timestamps $t_v$ in any individual state in $s_{next}$ (i.e., \{3,5\}) are added to $T$ for future planning.
    
In the next iteration of LSRP, the planning timestamp is $t_{min} = 3$ and $I_{curr}$ is $\{2,3\}$, as agents $i=2,3$ ends their previous actions at $t = 3$.
To plan the next action of agents $2,3$, LSRP retrieves the move action from $\Phi$ and set $s_{next}^2=(D,B,3,5)$ for agent $2$, which is $s^2_{next}$ shown in Fig.~\ref{fig:Lsrp}(c). 
Agent $3$ has no cached action in $\Phi$ and is planned by calling \asypibt.
Since agent $3$ is at its goal, \asypibt set agent $3$ wait in vertex C until $t_{next}$ by setting $s_{next}^3=(C,C,3,5)$ (Line \ref{alg2:nokexist}-\ref{alg2:waitend},Alg.\ref{lsrp:alg:asy_pibt}). 
All agents in $I_{curr}$ are now planned, and $S_T$,$T$ are updated.
LSRP ends this iteration and proceeds.
The third planning timestamp is $t_{min} = 5$ and $I_{curr}$ is $\{1,2,3\}$.
LSRP plans in a similar way,
and the movement of the agents is shown in Fig.~\ref{fig:Lsrp}(d).
The last timestamp is $t_{min} = 6$. In this iteration, in $s_{prev}$, all agents have reached their goals and LSRP terminates.

\begin{figure}[tb]
    \centering
    \includegraphics[width=.9\linewidth]{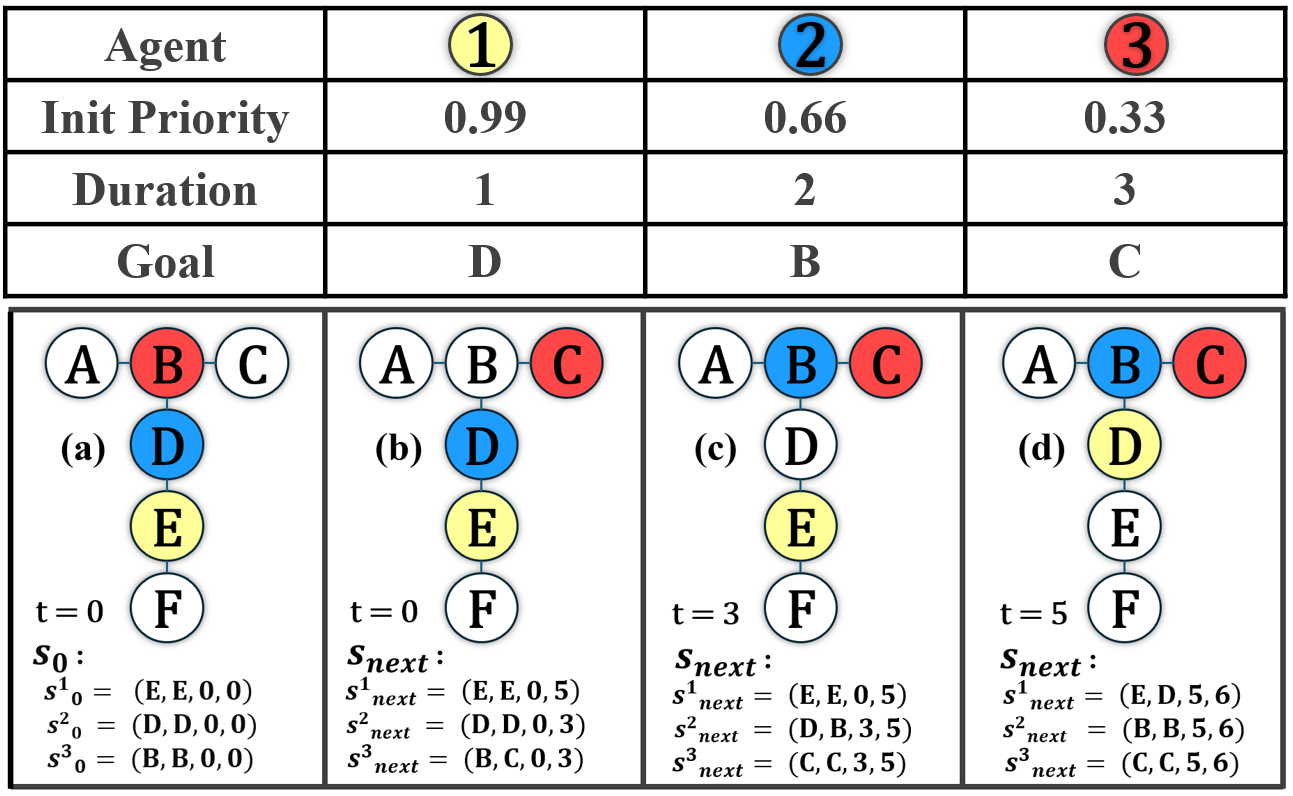}
    \caption{A toy example illustrating LSRP.}
    \label{fig:Lsrp}
\end{figure}

\subsection{Relationship to PIBT and Causal PIBT}\label{lsrp:sec:retopibt}
LSRP differs from PIBT in the following three aspects: \textbf{First}, PIBT solves MAPF where vertex and edge conflicts are considered. When one seeks to use PIBT-like approach to solve MAPF-DC (with duration conflict), the wait action may need to be considered in a similar way as LSRP does. Specifically, when two agents $i,j \in I_{curr}$ compete for the same vertex $v$, the higher-priority agent starts to move to $v$ until the lower-priority agent is pushed away from $v$ and reaches a less desired vertex $u$. Here, the wait and the move action of $i$ are planned together and the move action is cached in $\Phi$ for future execution. In PIBT, agent $i$ can move to $v$ that is currently occupied by $j$ as soon as $j$ leaves $v$, and there is no need for agent $i$ to wait and cache the move action.
\textbf{Second}, different from MAPF-DC, MAPF-AA has various durations. As a result, in each iteration (with planning timestamp $t$), LSRP plans for agents whose arrival time is equal to $t$, instead of planning all agents as PIBT does.
\textbf{Third}, LSRP also introduces a swap operation for MAPF-AA, which is demonstrated in Sec.~\ref{lsrp:sec:swap}.

Causal PIBT~\cite{okumura2021time} extends PIBT to handle delays caused by imperfect execution of the planned path for MAPF.
In the time-independent planning problem, due to the possible delays, the action duration is unknown until the action is finished. Causal PIBT thus plans ``passively'' by recording the dependency among the agents during execution using tree-like data structure. When one agent moves, Causal PIBT signals all the related agents based on their dependency. In MAPF-AA, the action duration is known, and LSRP thus plans ``actively'' by looking ahead into the future given the action duration of the agents, and caches the planned actions when needed. The fundamental ideas in LSRP and Causal PIBT are related and similar, but the algorithms are different since they are solving different problems.

\section{Analysis} \label{lsrp:analysis}
This section discusses the compromise on completeness for scalability in LSRP.
We borrow two concepts from Causal-PIBT~\cite{okumura2021time}:
\begin{itemize}
    \item \emph{Strong termination}: there is time point where all agents are at their goals.
    \item \emph{Weak termination}: all agents have reached their goals at least once. 
\end{itemize}
LSRP only guarantees \emph{weak termination} under conditions.
In this section, the words `reach goal' means that an agent reaches its goal location but may not stay there afterwards.

Given a graph $G$, a cycle $\{v_1,v_2,\dots,v_\ell,v_1\}$ is a special path that starts and ends at the same vertex $v_1$.
The length of a path (or a cycle) is the number of vertices in it.
Given a graph $G$ and $N$ agents, if there exists a cycle of length $\geq
N + 1$ for all pairs of adjacent vertices in $G$, then we call this graph a \emph{c-graph}.
Intuitively, for a c-graph, LSRP guarantees that the agent with highest priority can push away any other agents that block its way to its goal, and reaches its goal within finite time.
Once the agent reaches its goal, its priority is reset and thus becomes a small value, the priority of another agent becomes the highest and can move towards its goal.
As a result, all agents are able to reach their goals at a certain time.
Note that LSRP initializes all agents with a unique priority, and all agents' priorities are increased by one in each iteration of LSRP, unless the agent has reached its goal.
Let $i_{*} \in I$ denote the agent with the highest priority when initializing LSRP.
Let $D_{max}$ denote the largest duration for any agent and any edge. 

\begin{figure}[tb]
    \centering
    \includegraphics[width=0.7\linewidth]{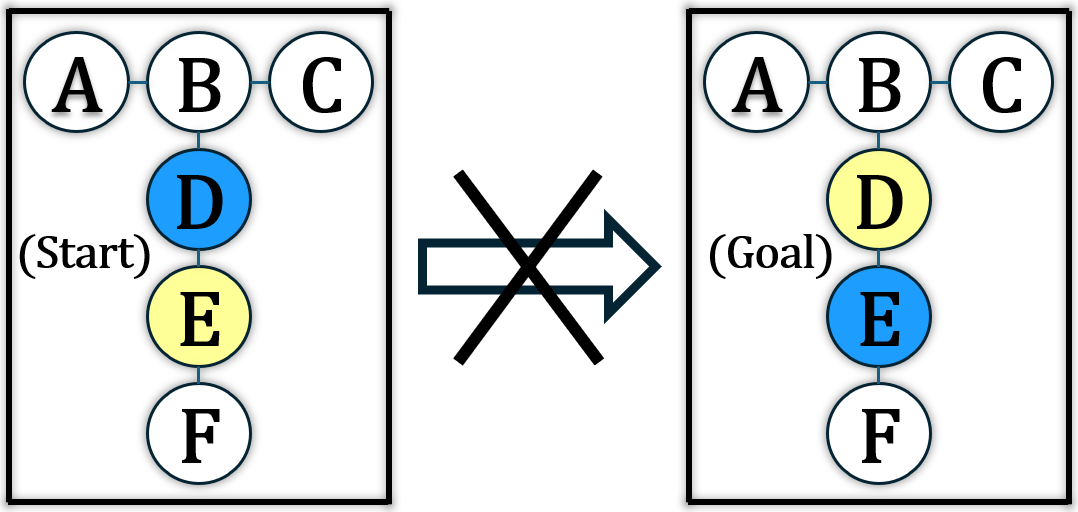}
    \caption{An example of LSRP never terminating with only \asypibt. 
    The blue and yellow agents will repeatedly push each other
    but will never successfully swap their locations to reach their goals.
    }
    \label{fig:failure}
\end{figure}

\begin{theorem}\label{lsrp:thm1}
In a c-graph, in LSRP, when $i_* \in I_{curr}$, let $v^*$ denote the nearest vertex from $v^{i_*}_g$ among all vertices in $C$, then $i_*$ can reach $v^*$ within time $N \cdot D_{max}$. 
\end{theorem}

\begin{proof}
    In each iteration, LSRP extracts the minimum timestamp $t_{min}$ from $T$, and at the end of the iteration, the newly added timestamps to $T$ must increase and be greater than $t_{min}$.
    As a result, $t_{min}$ keeps increasing as LSRP iterates and all agents are planned.
    Now, consider the iterations of LSRP where agent $i_* \in I_{curr}$ and is planned.
    In \asypibt, $v^*$ is check at first. If no other agent ($k$) occupies $v^*$, then $i_*$ reaches $v^*$ with a duration that is no larger than $D_{max}$.
    Otherwise (i.e., another agent $k$ occupies $v^*$), agent $i_*$ seeks to push $k$ to another vertices which may or may not be occupied by a third agent $k'$.
    Since the graph is a c-graph, there must be at least one unoccupied vertex in the cycle containing the $s^{i_*}_{prev}.v$, the current vertex occupied by $i_*$.
    In the worst case, all agents are inside this cycle and $i_*$ has to push all other agents before $i_*$ can reach $v_*$, which takes time at most $N\cdot D_{max}$, where the agent moves to its subsequent vertex in this cycle one after another.
\end{proof}

% We now consider the multi-agent case.
% Given a path $\pi=\{v_1,v_2,\dots,v_\ell\}$ in graph $G$, the length of $\pi$ is defined as the number of vertices in $\pi$.
Let $diam(G)$ denote the diameter of $G$, the length of the longest path between any pair of vertices in $G$.

\begin{theorem}\label{lsrp:thm2}
    For a c-graph, LSRP (Alg. \ref{lsrp:alg:lsrp}) returns a set of conflict-free paths such that for any agent $i \in I$ reaches its goal at a timestamp $t$ with $t \leq \text{diam}(G) \cdot N^2 \cdot D_{max}$.
\end{theorem}

\begin{proof}
From Theorem \ref{lsrp:thm1}, the agent with the highest priority arrives $v^*$ within $N \cdot D_{max}$.
So agent $i$ arrives $v_{g}^i$ within $\text{diam}(G) \cdot N \cdot D_{max}$.
Once agent $i$ reaches $v_{g}^i$, its priority is reset, which must be smaller than the priority of any other agents, and another agent $j$ gains the highest priority and is able to reach its goal.
This process continues until all agents in $I$ have gained the highest priority at least once, and reached their goals.
Thus, the total time for all agents to achieve their goals is within $\text{diam}(G) \cdot N^2 \cdot D_{max}$.
\end{proof}

\section{Extension with Swap Operation}\label{lsrp:sec:swap}

% Due to the compromise in completeness, LSRP may experience deadlock in MAPF-AA, which requires \emph{strong termination}.
Since LSRP only guarantees weak termination under conditions as aforementioned, for MAPF-AA in general, there are instances where LSRP never terminates although the instance is solvable.
Fig.~\ref{fig:failure} shows an example. 
Assume the yellow agent has higher initial priority than the blue agent. 
The yellow agent first pushes the blue agent until reaches goal $D$. 
Then, the blue agent has higher priority and pushes the yellow agent until reaches goal $E$. 
As a result, these two agents iteratively push each other and can never successfully swap their locations.

Similar issues also appear in PIBT, which is addressed by an additional swap operation~\cite{okumura2022priority}.
Inspired by this, we develop \asypibtswap(Alg.~\ref{lsrp:alg:asy_pibt}).
\asypibtswap takes the same input as \asypibt, and invokes a procedure \textsc{Swap-Possible-Required} to check 
if there exists an agent $j$ that needs to swap location with agent $i$.
If such a $j$ exists, \asypibtswap sorts the successor vertices in $C$ based on their distance to agent $j$'s goal from the furthest to the nearest (Line~\ref{alg2:creverse}, Alg.~\ref{lsrp:alg:lsrp}), and plans the actions of $i$ and $j$ differently:
if agent $i$ can move to the furthest vertex in $C$, and agent $j$ is not planned, then $j$ moves to $i$'s current vertex after $i$ vacates it (Line \ref{alg2:kexswap}-\ref{alg2:kexswapend},\ref{alg2:knexswap}-\ref{alg2:knexswapend} in Alg.~\ref{lsrp:alg:asy_pibt}).
We refer to this variant of LSRP with swap as LSRP-SWAP.
Since LSRP-SWAP alters the order of vertices in $C$, the analysis in Theorem~\ref{lsrp:thm1} based on the nearest vertex $v^*$ of an agent $i_*\in I_{curr}$ cannot be applied. 
Theorem~\ref{lsrp:thm1} and~\ref{lsrp:thm2} thus may not hold for LSRP-SWAP.

\textsc{Swap-Required-Possible} is similar to the concept of the swap operation in PIBT~\cite{okumura2023lacam2}.
The main difference is that LSRP-SWAP must consider the duration conflict between agents rather than the vertex or edge conflict in PIBT.
Appendix provides more detail of LSRP-SWAP.

\begin{figure}[tb]
    \centering
    \includegraphics[width=\linewidth]{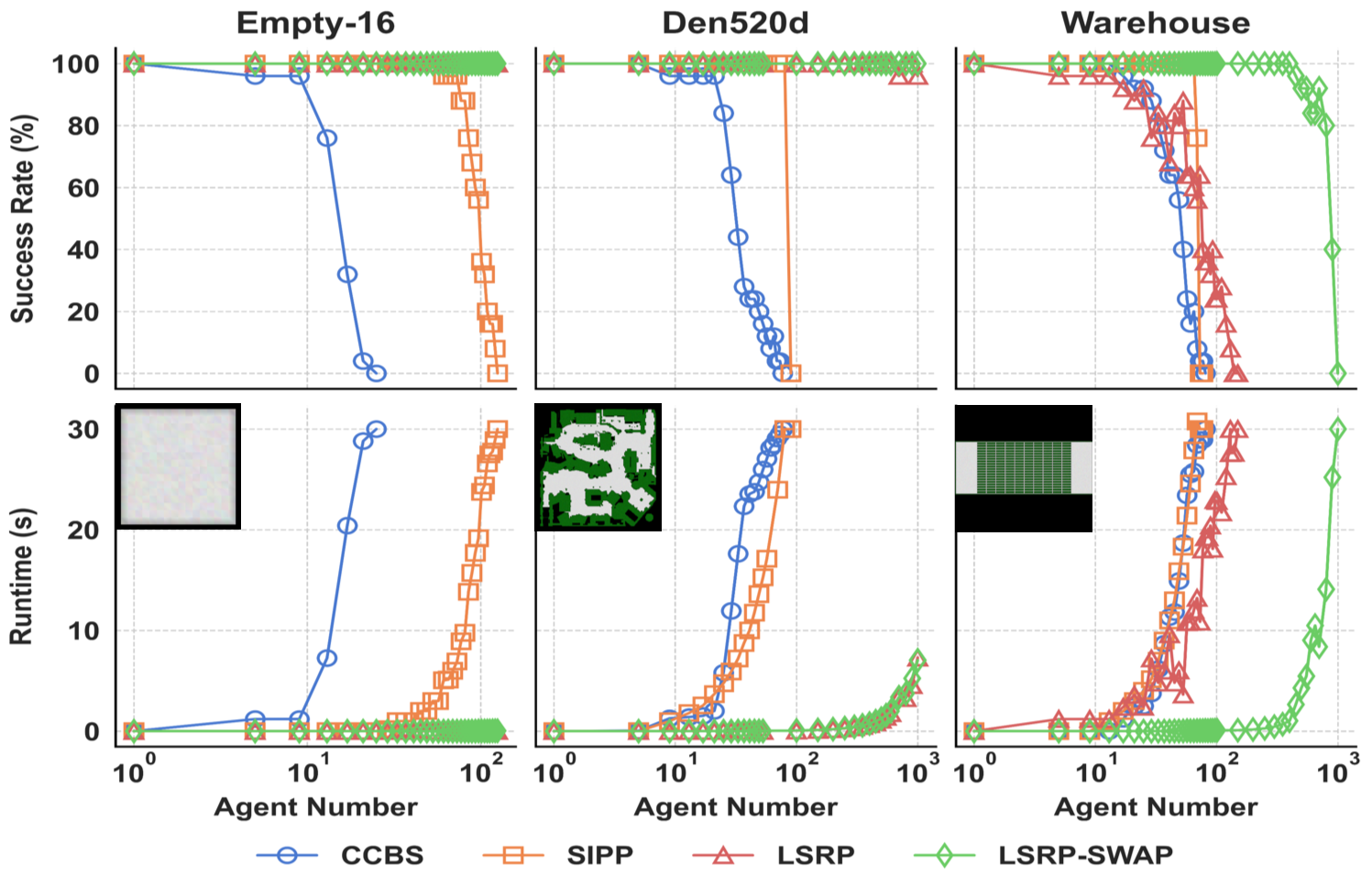}
    \caption{Success rate and runtime results}
    \label{fig:exp-scalability}
\end{figure}

\begin{figure}[tb]
    \centering
    \includegraphics[width=\linewidth]{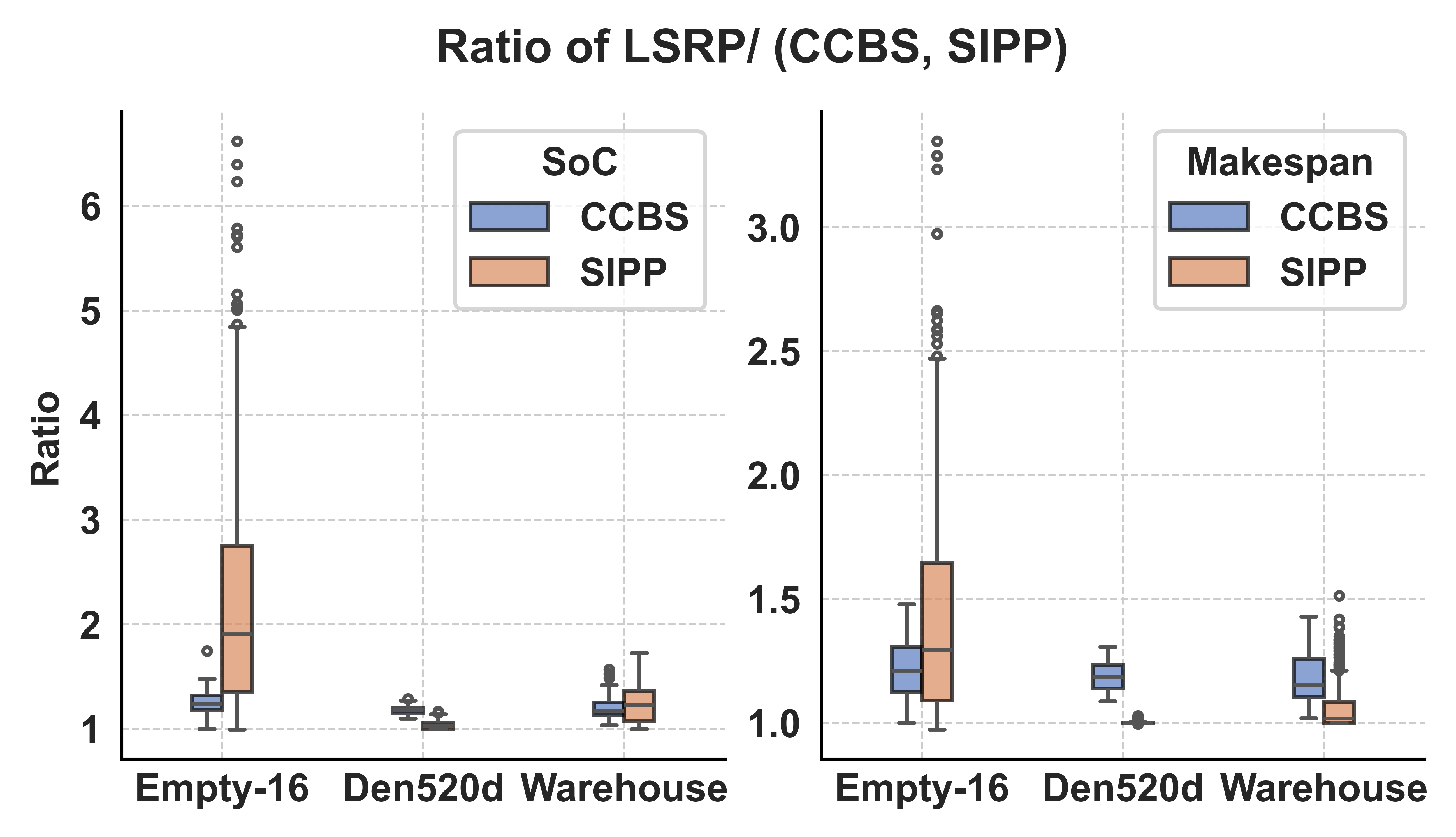}
    \caption{SoC and makespan ratios of LSRP}
    \label{fig:exp-quality_lsrp}
\end{figure}

\begin{figure}[tb]
    \centering
    \includegraphics[width=\linewidth]{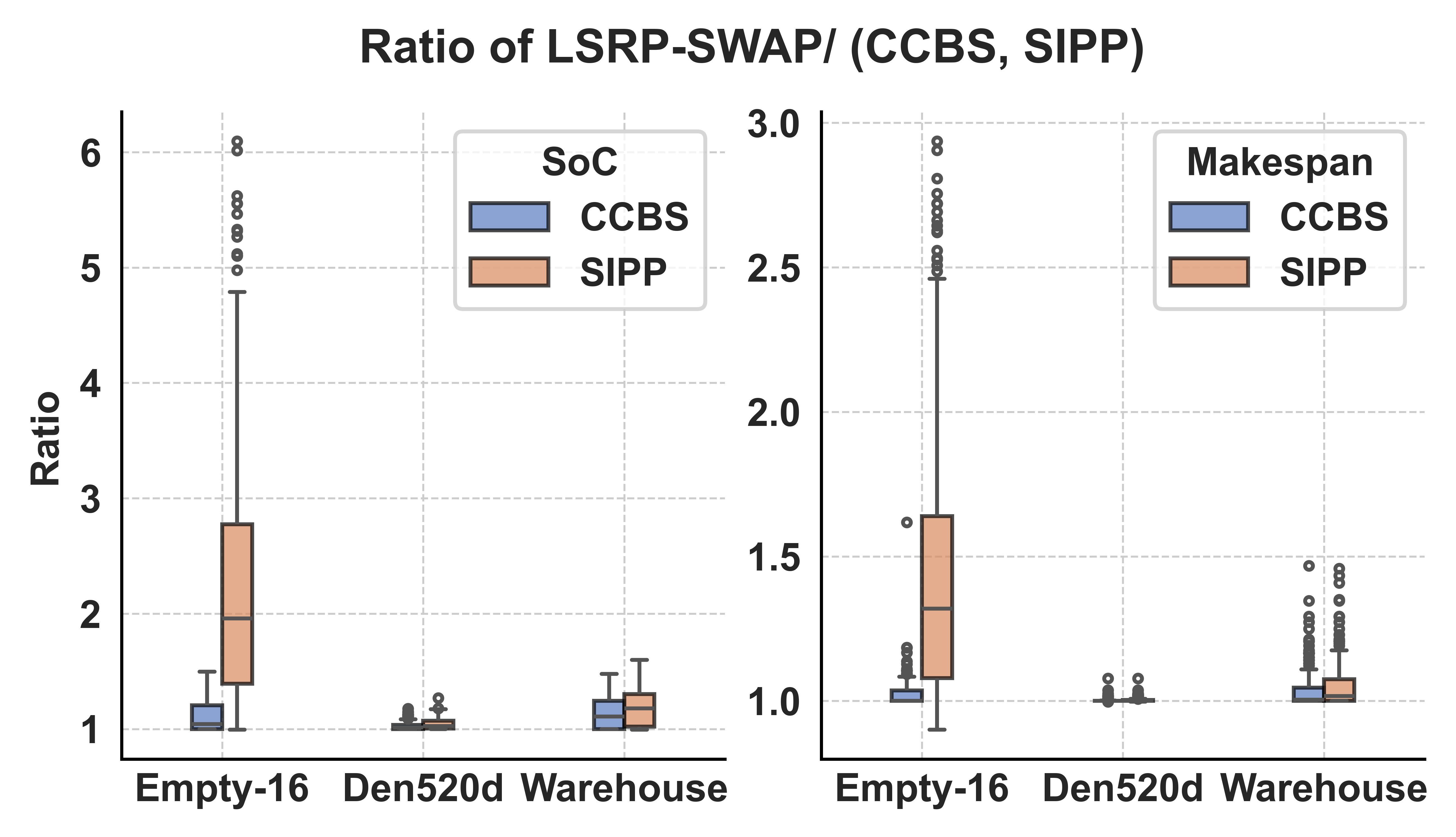}
    \caption{SoC and makespan ratios of LSRP-SWAP}
    \label{fig:exp-quality}
\end{figure}

\begin{figure}[tb]
    \centering
    \includegraphics[width=\linewidth]{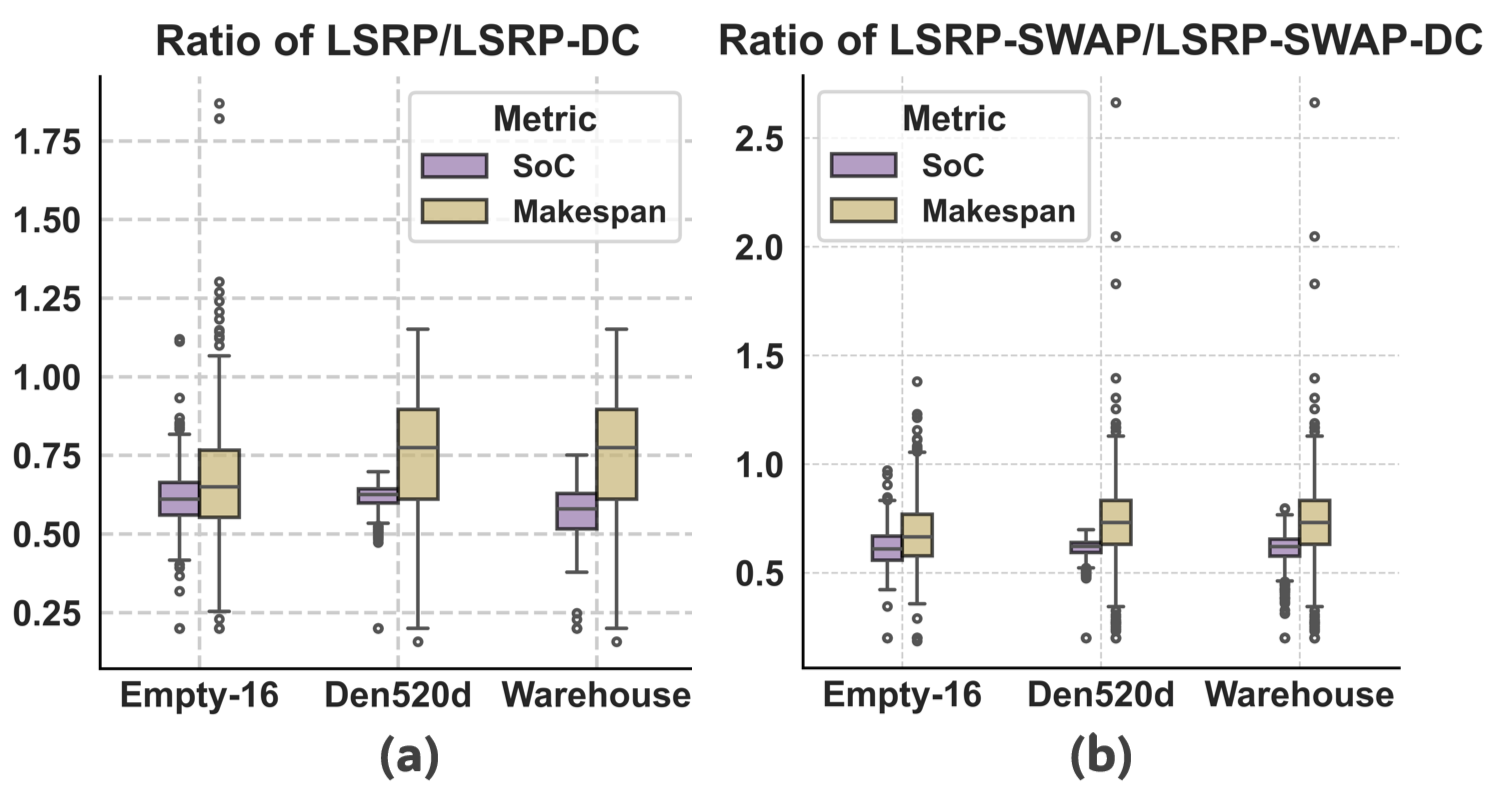}
    \caption{Solution costs comparison with and without considering asynchronous actions.}
    \label{fig:IAA}
\end{figure}

\section{Experimental Results} \label{lsrp:result}

Our experiments use three maps and the corresponding instances (starts and goals) from a MAPF benchmark~\cite{stern2019multi}.
For each map, we run 25 instances with varying number of agents $N$, and we set a 30-seconds runtime limit for each instance.
We make the grid maps four-connected, and each agent has a constant duration when going through any edge in the grid.
This duration constant of agents vary from 1.0 to 5.0.
We implement our LSRP and LSRP-SWAP in C++,
% \footnote{https://github.com/rap-lab-org/public\_LSRP}
and compare against two baselines.
The first baseline is a modified CCBS~\cite{ANDREYCHUK2022103662}.
The original CCBS implementation considers the shape of the agents and does not allow different agents to have different durations when going through the same edge.
We modified this public implementation by using the duration conflict and allow different agents to have different duration. Note that the constraints remain ``sound''~\cite{ANDREYCHUK2022103662}, and the solution obtained is optimal to MAPF-AA.
The second baseline \masipp (see Sec.~\ref{lsrp:other-related}) adopts prioritized planning using SIPP~\cite{phillips2011sipp} as the single-agent planner.
It uses the same initial priority as LSRP and LSRP-SWAP.
All tests use Intel i5 CPU with 16GB RAM.

\subsection{Success Rates and Runtime}\label{sec:sucrate}

Fig.~\ref{fig:exp-scalability} shows the success rates and runtime of the algorithms.
Overall, it is obvious that our LSRP and LSRP-SWAP can often handle more than an order of magnitude number of agents than the baselines within the time limit with much smaller runtime.
In particular, LSRP-SWAP runs fastest and handles up to 1000 agents in our tests.
In sparse maps (\texttt{Empty-16},\texttt{Den520d}), LSRP and LSRP-SWAP both scale well with respect to $N$, and the reason is that these graphs have many cycles that make the push operation highly efficient when resolving conflicts between the agents.
In a cluttered environment (\texttt{Warehouse}), LSRP has similar performance to \masipp while LSRP-SWAP outperforms both LSRP and \masipp, which shows the advantage of the swap operation in comparison with LSRP without swap. 

\subsection{Solution Quality}
Now we examine the solution quality of proposed methods.
We measure the solution quality using both sum of costs (SoC) and makespan.
We compare the solution quality (SoC and makespan) of LSRP and LSRP-SWAP with baselines (i.e., CCBS and \masipp) respectively. 
The ratios $A/B$ used for comparison are computed based on the instances from the experiment in Sec.~\ref{sec:sucrate} that were successfully solved by both planners A and B, where A is LSRP or LSRP-SWAP, and B is CCBS or \masipp.
The higher the ratio, the better the solution quality of the baseline $B$.

As shown in Fig.~\ref{fig:exp-quality_lsrp} and \ref{fig:exp-quality}, the median ratios is about 4x in SoC, and 1.25x in makespan, which means LSRP and LSRP-SWAP find more expensive solutions than the baselines.
It indicates LSRP and LSRP-SWAP achieve high scalability at the cost of solution quality, 
while the baselines usually find high quality solution with limited scalability.
Note that \masipp solves more instances than CCBS, and the ratios for CCBS and \masipp are thus calculated based on different sets of instances.
So \masipp sometimes has a higher ratio than the optimal planner CCBS on the map \texttt{Empty-16}.

\subsection{Impact of Asynchronous Actions}

This compares the solution costs of LSRP and LSRP-SWAP when asynchronous actions are considered versus when they are not.
We use the same instances as in Sec.~\ref{sec:sucrate}.
When ignoring the asynchronous actions, the duration constant of all agents are set to 5.0, and the resulting MAPF-AA problem becomes MAPF-DC, where all agents have common planning timestamps and can be planned in a step-by-step manner.
Note that even for MAPF-DC, due to the duration conflict, PIBT cannot be directly applied as discussed in Sec.~\ref{lsrp:sec:retopibt}.
As shown in Fig.~\ref{fig:IAA}, by considering the asynchronous actions, the obtained solutions are usually 30\% and sometimes up to 75\% cheaper than the solutions that ignore the asynchronous actions.
This result verifies the importance of considering asynchronous actions during planning, especially when agents have very different duration.

\section{Conclusion and Future Work} \label{lsrp:conclusion}
This paper develops rule-based planners for MAPF-AA by leveraging both PIBT for MAPF and LSS to handle asynchronous actions.
The experimental results verify their ability to achieve high scalability for up to a thousand agents in various maps, at the cost of solution quality.
Future work includes developing an anytime planner that can further improve solution quality within the runtime limit for MAPF-AA. Specifically, LSRP can be potentially extended to an anytime version that is similar to LaCAM over PIBT~\cite{okumura2023lacam2}. This extension would allow LSRP to iteratively optimize solution quality within the runtime limit.

\section*{Acknowledgements}
This work was supported in part by the Natural Science Foundation of China under Grant 62403313.

\bibliography{aaai25}
\section*{Appendix}
This section provides the details of LSRP-SWAP, which involves first detecting whether two agents need to swap their locations, and second planning the actions of the agents to achieve the swap.
LSRP-SWAP does not consider the duration conflict between the agents in the detection process.
When planning the actions of the agents to achieve the swap, the duration conflict is considered on Lines 17 and 25 in Alg. \ref{lsrp:alg:asy_pibt} to let the agents wait and then move when needed.

We introduce the following concepts, and some of them have been mentioned earlier in the main paper.

\noindent\textbf{SWAP}:
\swapop is introduced in~\cite{okumura2023lacam2}.
It is a process where two agents exchange their current vertices in a 
conflict-free manner through a sequence of actions.

\noindent\textbf{PUSH}:
When $i$ applies a \pushop operation on $j$,
$j$ moves to another adjacent vertex that is not currently occupied by $i$, 
then $i$ moves to $j$'s current vertex.

\noindent\textbf{PULL}:
This is the reverse operation of \pushop.
When $i$ applies a \pullop operation on $j$, $j$ moves to $i$'s current vertex after $i$ moves to another vertex.

\noindent\textbf{OCCUPANT} is a procedure that is used in Alg.~\ref{lsrp:alg:swap-required-possible} Line \ref{algI:joccupy}.
It takes $(v, I_{curr}, s_{prev}, s_{next})$ as the input, and seeks to find an agent $j\in I_{curr}$ that currently occupies $v$ and has not yet been planned, according to $s_{prev}$ and $s_{next}$. If no such an agent $j$ exists, \noindent\textbf{OCCUPANT} returns an empty set.

\section*{$\textsc{Swap-Required-Possible}$} 

The procedure $\textsc{Swap-Required-Possible}$ is called in Alg.~\ref{lsrp:alg:asy_pibt}.
Recall that $v_g^i$ denotes the goal vertex of agent $i \in I$, and $C$ denotes a set of vertices that agent $i$ can reach from its current vertex, sorted according to the distance from $v_g^i$ from the nearest to the furthest (Alg.~\ref{lsrp:alg:asy_pibt}).

Alg.~\ref{lsrp:alg:swap-required-possible} takes as input an agent $i \in I$ and a vertex $v$ in $C$ that is the closest to $v_g^i$.
Here, $s_{prev},s_{next},I_{curr}$ are ``global'' variables that are created on in Alg.~\ref{lsrp:alg:lsrp} at the beginning of each iteration of LSRP-SWAP.
When $v = v^i$ (Line \ref{algI:atgoal}), it guarantees that $v = v^i = v^i_g$, 
which means agent $i$ reaches the goal and $i$ does not need to swap vertices with others to get closer to $v^i_g$, 
and Alg.~\ref{lsrp:alg:swap-required-possible} ends.
Otherwise ($v\neq v^i$), it invokes \textsc{Occupant} to find an agent $j$ to swap with.

If \textsc{Occupant} finds an agent $j$ ($j \neq \emptyset$), it indicates that, agent $i$ can swap with agent $j$ to get closer to $v^i_g$.
Then, it invokes \textsc{Swap-Check} to check if it is enough to eventually swap agents $i$ and $j$ by just using a sequence of pull operations of agent $j$ to agent $i$ (Line \ref{algI:jswaprequired}).
If \textsc{Swap-Check} returns true, it means that $i,j$ cannot be swapped by just using a sequence of pull operations.
If \textsc{Swap-Check} returns false, it means that $i,j$ can be swapped by just using a sequence of pull operations, or agent $i,j$ do not need to be swapped at all.
On Line \ref{algI:jswaprequired}, when \textsc{Swap-Check} returns true, another \textsc{Swap-Check} is applied (Line \ref{algI:jswappossible}) to check if agent $i,j$ can be swapped by letting agent $i$ pull $j$. 
After \textsc{Swap-Check} at line \ref{algI:jswappossible} returns false, it means that using pull operations can swap the agents.
Agent $j$ is thus returned and the procedure ends (Line \ref{algI:jreturn}). 

\begin{algorithm}[tb]
\small
	\caption{\textsc{Swap-Required-Possible}}\label{lsrp:alg:swap}
	\begin{algorithmic}[1]
\Require $i,v$ 
\Notation $i,j,k \in I$
\State $v^i \leftarrow s_{prev}^i.v$
\State \textbf{if} $v = v^i {\rm \textbf{ then return }\emptyset}$ \label{algI:atgoal}
\State $j \leftarrow \textsc{Occupant} (v,I_{curr},s_{prev},s_{next})$ \label{algI:joccupy}
\If{$j \neq \emptyset$} 
\State $v^j \leftarrow s_{prev}^j.v$
\If{$\textsc{Swap-Check}(v^j,v^i)$} \label{algI:jswaprequired}
\If{$!\textsc{Swap-Check}(v^i,v^j)$} \label{algI:jswappossible}
\State \Return $j$ \label{algI:jreturn}
\EndIf
\EndIf
\EndIf
\For{$u \in \text{Neigh}(v^i)$} \label{algI:u}
\State $k \leftarrow \textsc{Occupant}(u,I_{curr},s_{prev},s_{next})$ \label{algI:koccupy}
\State \textbf{if} $k = \emptyset \vee s_{prev}^k.v = v \textbf{ then continue}$ \label{algI:atv}
\If{$\textsc{Swap-Check}(v^i,v)$} \label{algI:kswaprequired}
\If{$!\textsc{Swap-Check}(v,v^i)$} \label{algI:kswappossible}
\State \Return $k$ \label{algI:kreturn}
\EndIf
\EndIf
\EndFor
\State \Return $\emptyset$ \label{algI:end}
	\end{algorithmic}
\label{lsrp:alg:swap-required-possible}
\end{algorithm}

\begin{algorithm}[tb]
\caption{\textsc{Swap-Check}}
\begin{algorithmic}[1]
\Require $v^i,v^j$ 
\Notation $i,j \in I$
\State $v^{pl} \leftarrow v^{i};v^{pd}\leftarrow v^{j}$
\While{$v^{pl} \neq v^{j}$} 
\State $n \leftarrow (\text{Neigh}(v^{pl})).size()$ \label{algI:neighboursize}
\For{$u \in \text{Neigh}(v^{pl})$}
\If{$u = v^{pd}$}
\State $n-1$  \label{alg5:n-1}
\State ${\rm \textbf{continue}}$
\EndIf
\State $v^{pl} \leftarrow u$\label{algI:assignVertex}
\EndFor
\State \textbf{if} $n\, >= \,2 {\rm \textbf{ then return false }}$ \label{algIII:n2}
\State \textbf{if} $n\, <= \,0 {\rm \textbf{ then return true }}$ \label{algIII:n0}
\If{$v^{pd} = v^j_g$} \label{algIII:iatgoal}
\If{$argmin(h(i,u))_{u\in Neigh(v^{pl})} = v^{pd}$} \label{algIII:closest}
\State \Return ${\rm \textbf{true }}$ \label{algIII:failtoswap}
\EndIf
\EndIf
\State $v^{pd} \leftarrow v^{pl}$ \label{algI:whileend}
\EndWhile
\State \Return ${\rm \textbf{true } }$ \label{algIII:end}
\end{algorithmic}\label{lsrp:alg:possible-required}
\end{algorithm}

Subsequently, from Line \ref{algI:u} in Alg.~\ref{lsrp:alg:swap}, all adjacent vertices of $v^i$ are iterated in arbitrary order.
In each iteration, it tries to find an agent $k\in I_{curr}$ that currently occupies a neighbor vertex $u$ of $v^i$ and is not planned yet (Line \ref{algI:koccupy}).
If $u$ is unoccupied (i.e., $k = \emptyset$) or $u = v$ (i.e., $s^k_{prev}.v = v$), then these cases are already handled by Lines \ref{algI:joccupy}-\ref{algI:jreturn}, and the current iteration of the for loop ends (Line \ref{algI:atv}).
If $k$ is found, on Lines~\ref{algI:kswaprequired}-\ref{algI:kswappossible} in Alg.~\ref{lsrp:alg:swap}, \textsc{Swap-Required-Possible} first places agents $k$ and $i$ at vertices $v^i$ and $v$ respectively, and then invokes \textsc{Swap-Check} twice (Line \ref{algI:kswaprequired}-\ref{algI:kswappossible}) to check whether agent $k$ needs to swap with agent $i$.
% If first \textsc{Swap-Check} returns true and second \textsc{Swap-Check} returns false, 
The first call of \textsc{Swap-Check} checks whether agents $k$ and $i$ can swap their locations by repeatedly letting agent $k$ pulls agent $i$ (Line \ref{algI:kswaprequired}). If the \textsc{Swap-Check} on Line \ref{algI:kswaprequired} returns true, which indicates that swap operation cannot be done through agent $k$'s pulling. Then, the second \textsc{Swap-Check} is applied to check if these two agents can be swapped by letting agent $i$ pulls agent $k$ (Line \ref{algI:kswappossible}). When the second \textsc{Swap-Check} returns false, it indicates that the pull operations can swap agents $k$ and $i$. Thus agent $k$ is returned, indicating that agent $i$ needs to swap with agent $k$, and the procedure ends (Line \ref{algI:kreturn}). 

If no such an agent $j$ or $k$ exists, it indicates that either there is no need for agent $i$ to swap with another agent, or there is no way for agent $i$ to swap with another agent. Thus no agent is returned and Alg.~\ref{lsrp:alg:swap} ends (Line \ref{algI:end}).

\section*{\textsc{Swap-Check}}

\textsc{Swap-Check} is a procedure to predict if a sequence of pull operations can swap agent $i$ and $j$. 
In \textsc{Swap-Check}, let $v^{pl} \gets v^i$, where ``pl'' stands for pull, indicating that the agent attempts to pull another agent, and let $v^{pd} \gets v^j$, where ``pd'' stands for pulled, indicating that the agent is pulled by another agent. 
To do this, we continuously let $i$ pull $j$ (Line~\ref{algI:neighboursize}-\ref{algI:whileend}).
Specifically, in each while iteration, all neighbor vertices $u$ of the current vertex $v^{pl}$ of the pulling agent $i$ are considered, and $u$ is skipped if $u$ is same as the vertex $v^{pd}$ of the agent $j$ that is pulled.
If $u$ is different from $v^{pd}$, then $u$ is assigned to be the next vertex that the pulling agent $i$ will move to (Line~\ref{algI:assignVertex}).
Then, the pulled agent moves to the current vertex of the pulling agent (Line~\ref{algI:whileend}).
This while loop ends in four cases. 
\begin{itemize}
    \item 
(1) Agent $i$'s current vertex has at least 2 neighbor vertices, and $i$ can move to a vertex different from agent $j$'s current vertex. In this case, swap is possible through pull operations from agent $i$ to agent $j$, and no other operations are required. \textsc{Swap-Check} thus returns false (Line \ref{algIII:n2}). 
    \item 
(2) Agent $i$ has no neighbor vertices besides agent $j$'s current vertex. In this case, solely using push cannot achieve swap, and pull operation is thus required. 
\textsc{Swap-Check} thus returns true (Line \ref{algIII:n0}). 
    \item 
(3) Let $h(i,u)$ denote the shortest path distance from vertex $u$ to agent $i$'s goal $v^i_g$.
When the pulled agent $j$ is at its goal vertex $v_g^i$ (Line \ref{algIII:iatgoal}), if for the pulling agent $i$, the nearest vertex to its goal $v^i_g$ among the neighbor vertices Neigh($v^{pl}$) of $i$'s current vertex $v^{pl}$ is $v^{pd}$ (Line \ref{algIII:closest}), \textsc{Swap-Check} returns true, and swap operation is required.
    \item 
(4) When $v^{pl} = v^j$, it indicates that agent $i$ has pulled agent $j$ through a cycle (as defined in Sec. 5).
It means that the two agents have not swapped their vertices. Therefore, additional pull operation is required, and true is thus returned.
% (Line \ref{algIII:end}).
\end{itemize}

\begin{figure}[tb]
    \centering
     \includegraphics[width=0.8\linewidth]{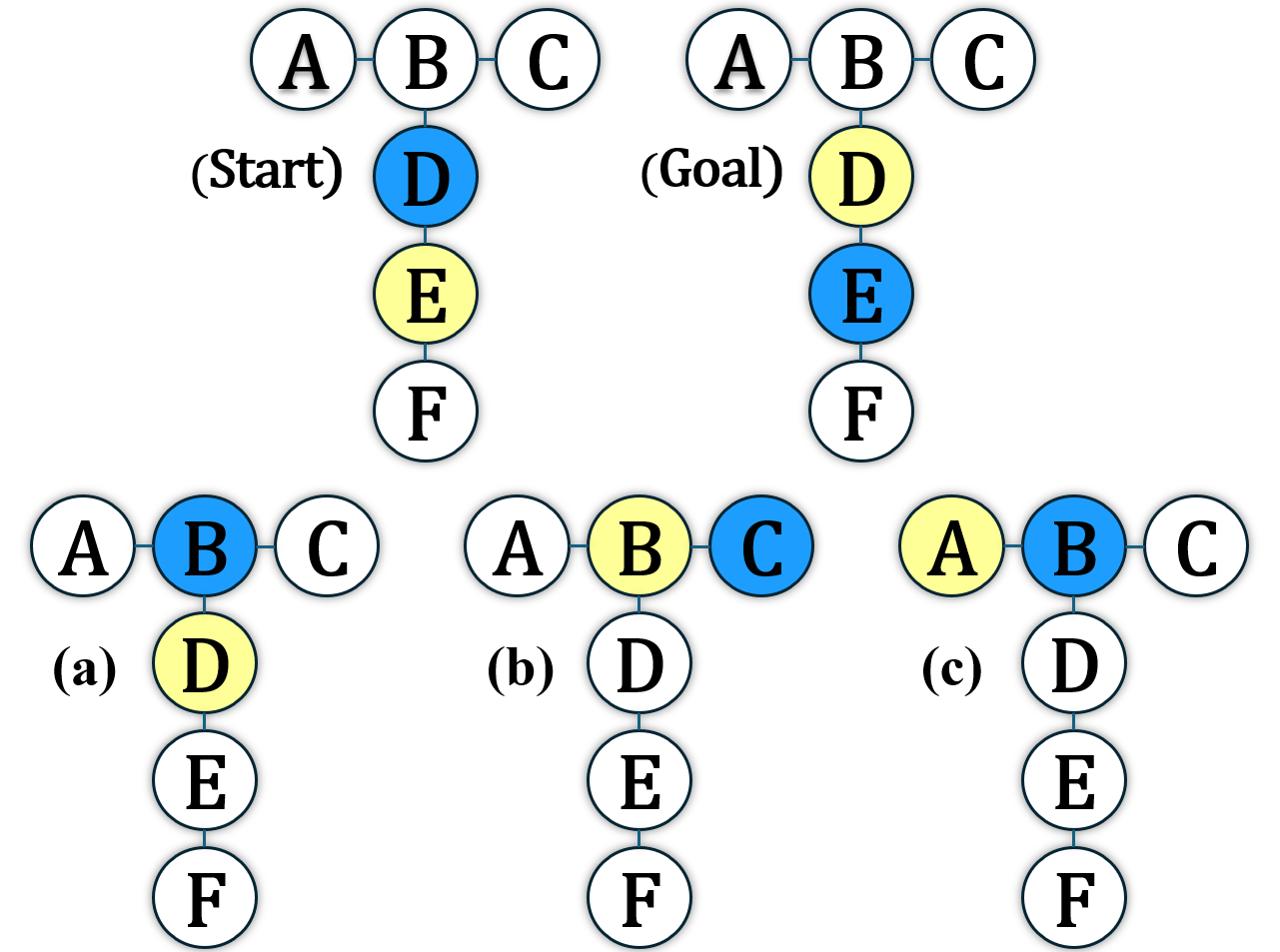}
    \caption{
    The two agents with color blue and yellow start like (Start), going through a series of collision-free operations, and finally reach (Goal). 
    The only feasible plan is to swap the location of two agents via operations (a), (b) and (c), where both agents need to move away from their goal locations.
    \asypibt cannot find a feasible solution in this case as it always lead to one of the agent moves towards to its goal location.}
    \label{figI:fail}
\end{figure}
\section*{Toy Example}
This section presents an example illustrating how LSRP-SWAP solves the instance in Fig.~\ref{figI:fail} that requires swapping the location of two agents in a tree-like graph.
Let $b$ denote the blue agent and $y$ denote the yellow agent. Initially, agent $y$ is assigned with the highest priority, while agent $b$ is assigned with the lowest priority.

\subsubsection{LSRP-SWAP for (a)}
As shown in Fig.~\ref{figI:fail}, starting from the initial state (Start), after a few push operations, LSRP-SWAP reaches (a).
In (a), agent $y$ is at its goal vertex, so its priority is reset, while agent $b$'s priority increases by 1. 
Now agent $b$ has the highest priority. LSRP-SWAP invokes \asypibtswap to plan for agent $b$ and $y$ in decreasing order of priority, and agent $b$ is planned at first. 

In the \asypibtswap procedure of $b$, it invokes \textsc{Swap-Required-Possible} with input $i = b, v = D$, since vertex D is the closest to $b$'s goal vertex (Line \ref{alg2:swap}, Alg. \ref{lsrp:alg:asy_pibt}). 
The procedure finds that agent $y$ occupies D through \textsc{Occupant} (Line \ref{algI:joccupy}, Alg. \ref{lsrp:alg:swap-required-possible}) and invokes \textsc{Swap-Check} (Line \ref{algI:jswaprequired}, Alg. \ref{lsrp:alg:swap-required-possible}).
Here, \textsc{Swap-Check} predicts whether the two agents can swap their locations by iteratively moving $b$ to $y$'s current vertex and moving $y$ to another vertex.
The input to this \textsc{Swap-Check} is that $v^i=B$ and $v^j=D$, meaning agent $y$ is at vertex $D$ and agent $b$ is at vertex $B$.
After a few iterations, $y$ arrives at vertex F and $b$ arrives at vertex E. $y$ finds no vertex to move to other than vertex E that is occupied by agent $b$ (Line \ref{algIII:n0}, Alg.~\ref{lsrp:alg:possible-required}).
Therefore, additional operation is required, and \textsc{Swap-Check} returns true.

Then, \textsc{Swap-Required-Possible} invokes another \textsc{Swap-Check} (Line \ref{algI:jswappossible}, Alg.\ref{lsrp:alg:swap-required-possible}) to predict whether two agents can swap their locations by iteratively moving $y$ to $b$'s current vertex and moving $b$ to another vertex.
The input to this \textsc{Swap-Check} is that agent $v^i=D$ and $v^j=B$, meaning agent $y$ is at vertex $D$ and agent $b$ is at vertex $B$.
\textsc{Swap-Check} finds that agent $b$'s current vertex B has two occupied vertices (A and C) that are different from $y$'s current vertex D, and agent $b$ can move to either of these two occupied vertices.
Swapping the location of both agents is thus possible since $B$ has two unoccupied neighbor vertices, and \textsc{Swap-Check} thus returns true (Line \ref{algIII:n2}, Alg. \ref{lsrp:alg:possible-required}). 

After \textsc{Swap-Check} returns true, \textsc{Swap-Required-Possible} ends and returns agent $y$ (Line~\ref{algI:jreturn}, Alg. \ref{lsrp:alg:swap-required-possible}), indicating that agent $y$ is able to swap vertices with agent $b$.
After \textsc{Swap-Required-Possible} returns agent $y$, the \asypibtswap procedure (Line~\ref{alg2:creverse}, Alg.~\ref{lsrp:alg:asy_pibt}) reverses the order of vertices in $C$ to start the swapping process, and $b$ moves to $C[0]$, which is the vertex C in this toy example.
Then, agent $b$ pulls $y$ to $b$'s current vertex (Line \ref{alg2:knexswap}-\ref{alg2:knexswapend}, Alg.~\ref{lsrp:alg:asy_pibt}) as shown in (b). 
So far, agents $b$ and $y$ are planned, and the LSRP-SWAP procedure for (a) in Fig.~\ref{figI:fail} ends.

\subsubsection{LSRP-SWAP for (b)}
In the next iteration of LSRP-SWAP, when planning for (b) in Fig.~\ref{figI:fail}, agent $b,y$ are both in $I_{curr}$ and are planned in decreasing order of priority. 
The procedure invokes \asypibtswap to plan for agent $b$ at first. 
Here, \asypibtswap seeks to move agent $b$ to vertex B and pushes agent $y$ away, thus a recursive call to \asypibtswap of agent $y$ is conducted. 

In this recursive call of \asypibtswap on agent $y$, 
it invokes the \textsc{Swap-Required-Possible} procedure with input $i = y, v = D$, since D is the closest vertex to agent $y$'s goal.
No agent occupies D, so the procedure arrives at Line \ref{algI:u} of Alg.~\ref{lsrp:alg:swap-required-possible}.
\textsc{Swap-Required-Possible} (Line \ref{algI:u} in Alg. \ref{lsrp:alg:swap-required-possible}) iterates all neighbor vertices of $y$'s current vertex B (in arbitrary order) .
For example, it checks vertex D and A and finds they are unoccupied, which means the \textsc{OCCUPANT} procedure (Line \ref{algI:atv} in Alg. \ref{lsrp:alg:swap-required-possible}) returns an empty set for vertices D and A.
Then, when iterating vertex $C$, \textsc{OCCUPANT} finds that agent $b$ occupies vertex C.
\textsc{Swap-Required-Possible} then predicts the case that agent $b$ moves to vertex B after agent $y$ moves to vertex D by invoking \textsc{Swap-Check} with input $v^i = B, v^j = D$ (Line \ref{algI:kswaprequired}, Alg.~\ref{lsrp:alg:swap-required-possible}).
In \textsc{Swap-Check}, after a few while-loop iterations, agent $y$ reaches vertex F and has no neighbor vertex to move to, swap is required and true is thus returned (Line \ref{algIII:n0}, Alg.~\ref{lsrp:alg:possible-required}). Then, \textsc{Swap-Possible} is invoked with input $v^i = D, v^j = B$. 
In this \textsc{Swap-Possible}, it finds that when agent $b$ is at vertex B, $b$ has 2 vertices to move to, so a swap is possible , true is thus returned (Line \ref{algIII:n2}, Alg.~\ref{lsrp:alg:possible-required}).
After \textsc{Swap-Possible}, agent $b$ is returned as the return argument $k$, the agent to swap vertices with agent $y$, and \textsc{Swap-Required-Possible} ends (Line \ref{algI:kreturn}, Alg.~\ref{lsrp:alg:swap-required-possible}).
In $y$'s \asypibtswap, its $C$ is reversed, and C[0] = vertex A.  
A is identified first; \asypibtswap finds that A is unoccupied, thus moving $y$ to A. 
Agent $b$ moves to B after $y$ arrives at A, as shown in (c).
Agents $b$ and $y$ are planned, and the LSRP-SWAP procedure for (b) ends.

\subsubsection{LSRP-SWAP for (c)}
Starting from (c), after a few while-loop iterations of LSRP-SWAP, where push operation is applied in each iteration, agents $b$ and $y$ eventually reach their goal vertices D and E, respectively. 
LSRP-SWAP successfully swaps the locations of the two agents.
\end{document}